\documentclass[12pt]{article}
\usepackage[latin1]{inputenc}
\usepackage[british]{babel}
\usepackage{cmap}
\usepackage{lmodern}

\usepackage{amssymb, amsmath, amsthm}
\usepackage[a4paper,top=25mm,bottom=25mm,left=25mm,right=25mm]{geometry}
\usepackage{adjustbox}
\usepackage{etex}
\usepackage{ragged2e}

\usepackage{authblk} 
\usepackage{pifont}
\usepackage{graphicx}
\usepackage[usenames,dvipsnames,svgnames,table]{xcolor}
\usepackage[figuresright]{rotating}
\usepackage{xtab} 
\usepackage{longtable} 
\usepackage{multirow}
\usepackage{footnote}
\usepackage[stable]{footmisc}
\usepackage{chngpage} 
\usepackage{pdflscape} 
\usepackage[nottoc,notlot,notlof]{tocbibind} 

\usepackage{pgfplots}
\pgfplotsset{compat=1.14}
\pgfplotsset{every tick label/.append style={font=\footnotesize}}
\usepgfplotslibrary{fillbetween}
\usepackage{setspace}

\makesavenoteenv{tabular}
\usepackage{tabularx}
\usepackage{booktabs}
\usepackage{threeparttable}
\usepackage[referable]{threeparttablex} 
\newcolumntype{R}{>{\raggedleft\arraybackslash}X}
\newcolumntype{L}{>{\raggedright\arraybackslash}X}
\newcolumntype{C}{>{\centering\arraybackslash}X}
\newcolumntype{M}[1]{>{\centering\arraybackslash}m{#1}}
\newcolumntype{K}{>{\columncolor{gray!20}}C}
\newcolumntype{k}{>{\columncolor{gray!20}}c}

\newlength{\tablen}

\usepackage{dcolumn} 
\newcolumntype{.}{D{.}{.}{-1}}

\usepackage{tikz}
\usetikzlibrary{arrows, calc, matrix, patterns, positioning, trees}
\usepackage[semicolon]{natbib}
\usepackage[hyphens]{url}
\usepackage{hyperref} 
\hypersetup{
  colorlinks   = true,    
  urlcolor     = blue,    
  linkcolor    = blue,    
  citecolor    = ForestGreen      
}
\usepackage{microtype}
\usepackage[justification=centerfirst]{caption} 

\usepackage[labelformat=simple,justification=centering]{subcaption}

\DeclareCaptionLabelFormat{parenthesis}{(#2)}
\makeatletter
\renewcommand\p@subfigure{\arabic{figure}.}
\makeatother

\DeclareCaptionLabelFormat{parenthesis}{(#2)}
\makeatletter
\renewcommand\p@subtable{\arabic{table}.}
\makeatother

\usepackage{enumitem}
\setlist[itemize]{leftmargin=2.5\parindent}
\setlist[enumerate]{leftmargin=2.5\parindent}

\def\addlegendimage{\csname pgfplots@addlegendimage\endcsname}

\theoremstyle{plain}

\newtheorem{proposition}{Proposition}

\theoremstyle{definition}

\newtheorem{example}{Example}[section]

\theoremstyle{remark}

\makeatletter
\let\@fnsymbol\@alph
\makeatother

\def\keywords{\vspace{.5em} 
{\noindent \textit{Keywords}: }}

\def\JEL{\vspace{.5em} 
{\noindent \textbf{\emph{JEL} classification number}: }}

\def\AMS{\vspace{.5em} 
{\noindent \textbf{\emph{MSC} class}: }}

\author{
\href{https://sites.google.com/view/laszlocsato}{L\'aszl\'o Csat\'o}\thanks{~Corresponding author. E-mail: \emph{laszlo.csato@sztaki.hu} \newline HUN-REN Institute for Computer Science and Control (HUN-REN SZTAKI), Laboratory on Engineering and Management Intelligence, Research Group of Operations Research and Decision Systems, Budapest, Hungary \newline 
Corvinus University of Budapest (BCE), Institute of Operations and Decision Sciences, Department of Operations Research and Actuarial Sciences, Budapest, Hungary}
$\qquad \qquad$
\href{https://sites.google.com/view/doragretapetroczy}{D\'ora Gr\'eta Petr\'oczy}\thanks{~E-mail: \emph{petroczy.dora@nje.hu} \newline
Institute of Economic, Centre for Economic and Regional Studies (KRTK), E\"otv\"os Lor\'and Research Network (ELKH), Budapest, Hungary \newline
MNB Institute, John von Neumann University, Budapest, Hungary}
}

\title{Bibliometric indices as a measure of performance \\ and competitive balance in the knockout stage of \\ the UEFA Champions League}
\date{\today}

\def\Dedication{ 
{\noindent $\mathfrak{Nur}$ $\mathfrak{wer}$ $\mathfrak{mit}$ $\mathfrak{geringen}$ $\mathfrak{Mitteln}$ $\mathfrak{Großes}$ $\mathfrak{tut}$, $\mathfrak{hat}$ $\mathfrak{es}$ $\mathfrak{gl\ddot{u}cklich}$ $\mathfrak{getroffen}$.}\footnote{~``\emph{Only he that does great things with small means has made a successful hit.}'' (Source: Carl von Clausewitz: \emph{On War}, Book 7, Chapter 22---On the Culminating Point of Victory, translated by Colonel James John Graham, London, N. Tr\"ubner, 1873. \url{http://clausewitz.com/readings/OnWar1873/TOC.htm})}

\flushright
\noindent (Carl von Clausewitz: \emph{Vom Kriege})

\vspace{1cm} 
\justify }

\begin{document}
\newgeometry{top=20mm,bottom=20mm,left=25mm,right=25mm}

\maketitle
\thispagestyle{empty}
\Dedication

\begin{abstract}
\noindent
We argue for the application of bibliometric indices to quantify the long-term uncertainty of outcome in sports. The Euclidean index is proposed to reward quality over quantity, while the rectangle index can be an appropriate measure of core performance. Their differences are highlighted through an axiomatic analysis and several examples. Our approach also requires a weighting scheme to compare different achievements. The methodology is illustrated by studying the knockout stage of the UEFA Champions League in the 20 seasons played between 2003 and 2023: club and country performances as well as three types of competitive balance are considered. Measuring competition at the level of national associations is a novelty. All results are remarkably robust concerning the bibliometric index and the assigned weights. 
Since the performances of national associations are more stable than the results of individual clubs, it would be better to build the seeding in the UEFA Champions League group stage upon association coefficients adjusted for league finishing positions rather than club coefficients.


\AMS{91A80, 91B82}

\JEL{C43, D40, Z20}
\end{abstract}

\clearpage

\section{Introduction} \label{Sec1}

This paper aims to connect two seemingly distant fields of academic research: scientometrics and sports economics. However, the approach is not unique, \citet{KoczyStrobel2010} and \citet{LehmannWohlrabe2017}  actually do the reverse by applying sports ranking methods to journals.

Following the pioneering work of \citet{Garfield1979}, a plethora of measures have been suggested, tested, and debated to quantify the influence of a scholar \citep{RavallionWagstaff2011, GlanzelMoed2013, WildgaardSchneiderLarsen2014}. Most academics have probably met with the $h$-index during the evaluation of grant proposals. Nonetheless, there is no consensus on what is the best bibliometric index, mainly due to the need to balance quantity (represented by the number of papers written) and quality (represented by the number of citations to those papers). This dilemma has motivated several articles on the axiomatisation of scientometric measures which can greatly help our understanding of their characteristics \citep{Marchant2009, BouyssouMarchant2010, KoczyNichifor2013, BouyssouMarchant2014, delaVegaVolij2018}.

The economic analysis of team sports is often centred around the issue of competitive balance since the celebrated article of \citet{Rottenberg1956}. Any sporting contest establishes a hierarchy among the competitors by differentiating between winners and losers.
Therefore, uncertainty of outcome has at least three different dimensions \citep{BuzzacchiSzymanskiValletti2003}:
(a) short-term (whether the ex ante winning probabilities of the playing teams are equal in a match);
(b) medium-term (whether all teams have the same chance to be the champion in a given season);
(c) long-term (how volatile is the relative standing of different clubs across some seasons).
Obviously, the measurement of competitive balance is much more than a simple academic problem. For example, several antitrust decisions are based on the claim that fan interest decreases when the uncertainty of outcome declines, thus, public interest demands to maintain competitive balance \citep{Szymanski2003}.

While there are numerous attempts to quantify the level of competition in sports leagues \citep{Zimbalist2002, ManasisAvgerinouNtzoufrasReade2013, ManasisNtzoufras2014}, there exists much less literature on how to measure it in knockout tournaments \citep{delCorral2009, ConsidineGallagher2018}. However, several popular sports championships---the tennis Grand Slam tournaments, the final stages of the FIFA World Cups and UEFA European Championships in (association) football, the playoffs in the NBA, NFL, NHL---are designed in this format. A natural expectation of both the fans and the regulators can be to avoid the dominance of the same players or teams across several seasons.

The current paper aims to quantify long-term uncertainty of outcome in such tournaments. The main challenge is similar to the central problem of bibliometrics: What is the trade-off between quantity and quality? Does a team obtain a greater market share when it has won the title once or when it has never won but qualified for the semifinals four times? A further difficulty of the measurement can be that the teams participating in a given year may not participate the year after \citep{Koning2009}.

We offer a solution to these challenges by using bibliometric indices. In particular, the Euclidean index \citep{PerryReny2016} and the rectangle index \citep{FennerHarrisLeveneBar-Ilan2018, LeveneFennerBar-Ilan2019} will be considered with various weights assigned for the achievements of the teams, analogous to the number of times a paper has been referred to. An axiomatic comparison will be provided to highlight the differences between the two measures. The Euclidean index rewards top achievements and is always increased when the performance improves. On the other hand, the rectangle index focuses on core performance. For example, if a team usually qualifies for the semifinals, it does not count whether it is eliminated in the Round of 16 or the quarterfinals in the next year. Since both concepts can be justified, we do not promote any of them over the other.

Several papers quantify uncertainty of outcome between consecutive seasons by using sophisticated statistical techniques. While our approach is much simpler, that can help its spread among the stakeholders as the popularity of scientometric indices shows: it is usually more attractive---although might be potentially misleading---to express a complex phenomenon in sports (or any other field) with a single number. The suggested framework can be understood with limited methodological knowledge, and one can easily ``open the black box'' behind the results. Anyway, measuring trends over time in a novel manner compared to the prior literature is able to reinforce or deny well-established findings which is an important issue in academic research, even though the conclusions of the analysis do not fundamentally differ from the findings of the ``standard'' literature.

The theoretical overview is followed by a real-world application. The proposed methodology will be used first to quantify the performance of clubs and national associations in the recent seasons of the UEFA Champions League, the most prestigious European football tournament.
In particular, a rolling five-season window is chosen because both the UEFA club and country coefficients---that measure the results of individual clubs (for seeding in the Champions League) and the member associations (for determining the access list of the Champions League), respectively---take into account the matches played over the past five seasons of European competitions. \citet{Milanovic2005} also examines five-year periods as reasonable over which concentration can be observed and calculated.

After that, three types of competitive balance will be analysed at the level of clubs, between the countries, and within the countries. It seems that previous studies have not discussed competition at the level of national associations in the UEFA Champions League. Therefore, they have failed to identify if different clubs from the same country have dominated the tournament. We can grab even this aspect of reduced competition.


The rest of the paper is organised as follows.
Section~\ref{Sec2} presents the two bibliometric indices and their properties. The UEFA Champions League data and the details of our application are discussed in Section~\ref{Sec3}, followed by a concise overview of previous literature on competitive balance in European football and the analysis of the results in Section~\ref{Sec4}. Finally, Section~\ref{Sec5} concludes.

\section{Methodology: bibliometric indices} \label{Sec2}

This section presents the axiomatic comparison of two bibliometric indices that will be used to quantify competitive balance. While it does not contain any new results, it might be useful for readers who are not familiar with the principles behind these measures.

Consider a championship consisting of $n$ races, where each position in each race earns some non-negative points for the competitors. The points scored by a contestant are collected into the \emph{score vector} $\mathbf{x} = \left[ x_1, x_2, \dots , x_n \right]$ that is sorted in descending order, i.e.\ $x_i \geq x_j$ for all $1 \leq i < j \leq n$. The achievement of a contestant in the whole championship is measured by an \emph{aggregation function} $f$ that maps score vectors to the set of non-negative real numbers.

The score vector is equivalent to a citation vector of a scholar, consequently, a bibliometric index can be used for its evaluation. In the case of competitive balance, most papers use the simple aggregation function $f(\mathbf{x}) = \sum_{i=1}^n x_i$, essentially counting the number of points for each competitor. However, this is not necessarily an appropriate approach in a knockout tournament. We consider two other bibliometric indices, in particular, the \emph{Euclidean index} \citep{PerryReny2016} and the \emph{rectangle index} \citep{FennerHarrisLeveneBar-Ilan2018, LeveneFennerBar-Ilan2019}.

The Euclidean index $E(\mathbf{x})$ of a score vector $\mathbf{x} = \left[ x_1, x_2, \dots , x_n \right]$ is the Euclidean norm of the score vector:
\[
E \left( \mathbf{x} \right) = \sqrt{\sum_{i=1}^n x_i^2}.
\]

The rectangle index $R(\mathbf{x})$ of a score vector $\mathbf{x} = \left[ x_1, x_2, \dots , x_n \right]$ is the area of the largest rectangle that can fit under the score vector:
\[
R \left( \mathbf{x} \right) = \max_{1 \leq i \leq n} i x_i.
\]
If the largest rectangle is a square, the rectangle index is equivalent to the well-known $h$-index \citep{Hirsch2005}, that is, the maximal number $h$ of races where the contestant scored at least $h$ points in each of them.

In order to better understand these measures, it is worth using an axiomatic approach.
Some reasonable requirements for aggregation functions are the following:
\begin{itemize}
\item
\emph{Monotonicity}: if score vector $\mathbf{x} = \left[ x_1, x_2, \dots , x_n \right]$ is dominated by score vector $\mathbf{y} = \left[ y_1, y_2, \dots , y_n \right]$, that is, $x_i \leq y_i$ holds for all $1 \leq i \leq n$, then $f(\mathbf{x}) \leq f(\mathbf{y})$.
\item
\emph{Independence}: if $f(\mathbf{x}) \leq f(\mathbf{y})$ and a new, $n+1$th race is added to the championship where both competitors score the same number of points, then the ordinal ranking of their achievements does not change.
\item
\emph{Depth relevance}: if score vector $\mathbf{x}$ is modified such that a new, $n+1$th race is added and the number of points $x_i$ in the $i$th race is split into two parts between the $i$th and $n+1$th races, then the overall achievement should decrease.
\item
\emph{Scale invariance}: $f(\mathbf{x}) \leq f(\mathbf{y}) \iff f(c\mathbf{x}) \leq f(c\mathbf{y})$ for any $c > 0$.
\item
\emph{Directional consistency}: if $f(\mathbf{x}) = f(\mathbf{y})$ holds for the score vectors $\mathbf{x}$ and $\mathbf{y}$, and both of them are shifted by the same score vector $\mathbf{d}$ such that $f(\mathbf{x} + \mathbf{d}) = f(\mathbf{y} + \mathbf{d})$, then $f(\mathbf{x} + \lambda \mathbf{d}) = f(\mathbf{y} + \lambda \mathbf{d})$ for any $\lambda > 1$.
\item
\emph{Uniform citation}: if the score vector $\mathbf{x}$ is uniform, that is, $x_1 = x_2 = \cdots = x_n$, then $f(\mathbf{x}) = \sum_{i=1}^n x_i$.
\item
\emph{Uniform equivalence}: for any score vector $\mathbf{x}$, there exists a uniform score vector $\mathbf{u}$ dominated by $\mathbf{x}$ such that $f(\mathbf{x}) = f(\mathbf{u})$.
\end{itemize}
Monotonicity is a natural requirement, the achievement of a contestant with at least the same performance cannot be lower.
Independence allows the addition of identical records for two competitors without changing their ranking.
Depth relevance excludes the overall achievement to be maximised by spreading the number of points collected thinly across the races. In other words, this property rewards quality over quantity. Note that summing up the number of points does not satisfy depth relevance.
According to scale invariance, multiplying the number of points by a positive scaling factor does not affect the order of two contestants. The $h$-index does not satisfy scale invariance \citep{PerryReny2016}.
Directional consistency means that if two tied competitors remain tied when their number of points are increased by the same amounts in decreasing order, they continue to remain equally ranked if this growth vector is multiplied by any positive constant.
Uniform citation requires the overall achievement to be equal to the number of points when they are level across all races.
Finally, according to uniform equivalence, the same number of points should be counted for each race in the overall achievement of a contestant.

An aggregation function $f$ satisfies monotonicity, independence, depth relevance, scale invariance, and directional consistency if and only if it is equivalent to the Euclidean index \citep[Theorem~1]{PerryReny2016}.
An aggregation function $f$ meets monotonicity, uniform citation, and uniform equivalence if and only if it is the rectangle index \citep[Theorem~4.1]{LeveneFennerBar-Ilan2019}.

A stronger version of scale invariance is satisfied by both measures.

\begin{proposition} \label{Prop1}
The share of a competitor's achievement in the total achievement under the Euclidean and the rectangle indices does not change if the scores are multiplied by the same positive number $c > 0$.
\end{proposition}

\begin{proof}
The invariance is implied by $E(c \mathbf{x}) = c E(\mathbf{x})$ and $R(c \mathbf{x}) = c R(\mathbf{x})$ for any score vector $\mathbf{x}$.
\end{proof}

According to Proposition~\ref{Prop1}, the two approaches are influenced only by the ratio of the points to be scored in each race, but they do not depend on the absolute size of the points.

Consider a \emph{binary championship} where there are only two types of positions, for example, a contestant either qualifies or does not qualify.

\begin{proposition} \label{Prop2}
The rankings of the competitors under the Euclidean and rectangle indices coincide in a binary championship.
\end{proposition}

\begin{proof}
Due to Proposition~\ref{Prop1}, it can be assumed without loss of generality that the points to be scored are $0$ and $1$. Consequently, the score vector of any competitor consists of $k$ pieces of one and $n-k$ pieces of zero. Then, $E(\mathbf{x}) \leq E(\mathbf{y})$ if and only if score vector $\mathbf{x}$ does not contain more ones than score vector $\mathbf{y}$. The same argument can be applied for the rectangle index.
\end{proof}

The differences between the Euclidean and rectangle indices are highlighted by the fact that the latter measure fails to satisfy three properties used to characterise the former as the following examples show. However, the rectangle index remains scale invariant (Proposition~\ref{Prop1}) which is the main reason why we have chosen it instead of the $h$-index.

\begin{figure}[t!]
\centering

\begin{subfigure}{\textwidth}
\centering
\caption{The original score vectors: Player 1 is better than Player 2 \\ by the Euclidean index but worse by the rectangle index}
\label{Fig1a}

\begin{tikzpicture}
\begin{axis}[width=0.3\textwidth, 
height=0.35\textwidth,
tick label style={/pgf/number format/fixed},
symbolic x coords={($T1$)},
xtick = \empty,
xlabel = Player 1,
xlabel style = {font=\small},
ybar = 0.25cm,
ymax = 5.5,
ymin = 0,
ymajorgrids = true,
bar width = 1cm,
]

\addplot [red, pattern color = red, pattern = horizontal lines, very thick] coordinates{
(($T1$),5)
};

\addplot [red, pattern color = red, pattern = horizontal lines, very thick] coordinates{
(($T1$),1)
};
\end{axis}
\end{tikzpicture}
\hspace{0.2\textwidth}
\begin{tikzpicture}
\begin{axis}[width=0.3\textwidth, 
height=0.35\textwidth,
tick label style={/pgf/number format/fixed},
symbolic x coords={($T1$)},
xtick = \empty,
xlabel = Player 2,
xlabel style = {font=\small},
ybar = 0.25cm,
ymax = 5.5,
ymin = 0,
ymajorgrids = true,
bar width = 1cm,
]

\addplot [blue, pattern color = blue, pattern = dots, very thick] coordinates{
(($T1$),3)
};

\addplot [blue, pattern color = blue, pattern = dots, very thick] coordinates{
(($T1$),3)
};
\end{axis}
\end{tikzpicture}
\end{subfigure}

\vspace{0.2cm}
\begin{subfigure}{\textwidth}
\centering
\caption{The modified score vectors: Player 1 is better than Player 2 by both \\ indices although both players scored five points in the additional race}
\label{Fig1b}

\begin{tikzpicture}
\begin{axis}[width=0.4\textwidth, 
height=0.35\textwidth,
tick label style={/pgf/number format/fixed},
symbolic x coords={($T1$)},
xtick = \empty,
xlabel = Player 1,
xlabel style = {font=\small},
ybar = 0.25cm,
ymax = 5.5,
ymin = 0,
ymajorgrids = true,
bar width = 1cm,
]

\addplot [red, pattern color = red, pattern = horizontal lines, very thick] coordinates{
(($T1$),5)
};

\addplot [red, pattern color = red, pattern = horizontal lines, very thick] coordinates{
(($T1$),5)
};

\addplot [red, pattern color = red, pattern = horizontal lines, very thick] coordinates{
(($T1$),1)
};
\end{axis}
\end{tikzpicture}
\hspace{0.1\textwidth}
\begin{tikzpicture}
\begin{axis}[width=0.4\textwidth, 
height=0.35\textwidth,
tick label style={/pgf/number format/fixed},
symbolic x coords={($T1$)},
xtick = \empty,
xlabel = Player 2,
xlabel style = {font=\small},
ybar = 0.25cm,
ymax = 5.5,
ymin = 0,
ymajorgrids = true,
bar width = 1cm,
]

\addplot [blue, pattern color = blue, pattern = dots, very thick] coordinates{
(($T1$),5)
};

\addplot [blue, pattern color = blue, pattern = dots, very thick] coordinates{
(($T1$),3)
};

\addplot [blue, pattern color = blue, pattern = dots, very thick] coordinates{
(($T1$),3)
};
\end{axis}
\end{tikzpicture}
\end{subfigure}

\captionsetup{justification=centering}
\caption{The rectangle index does not satisfy independence}
\label{Fig1}

\end{figure}


\begin{example} \label{Examp21}
\emph{The rectangle index does not satisfy independence} \\
Take the score vectors $\mathbf{x} = \left[ 5, 1 \right]$ and $\mathbf{y} = \left[ 3, 3 \right]$ of Players $1$ and $2$, plotted in Figure~\ref{Fig1a}. Since $E(\mathbf{x}) = \sqrt{26} \approx 5.1$ and $E(\mathbf{y}) = \sqrt{18} \approx 4.24$, the Euclidean index favours Player $1$. Contrarily, $R(\mathbf{x}) = 1 \times 5 = 5$ and $R(\mathbf{y}) = 2 \times 3 = 6$, thus, the rectangle index prefers Player $2$ to Player $1$.

Consider a new, third race where both players score five points, that is, $\mathbf{x}' = \left[ 5, 5, 1 \right]$ and $\mathbf{y}' = \left[ 5, 3, 3 \right]$, as depicted in Figure~\ref{Fig1b}. Now,
\[
E \left( \mathbf{x}' \right) = \sqrt{51} \approx 7.14 > E \left( \mathbf{y}' \right) = \sqrt{43} \approx 6.56 \text{, and}
\]
\[
R \left( \mathbf{x}' \right) = 2 \times 5 = 10 > R \left( \mathbf{y}' \right) = 3 \times 3 = 9,
\]
hence, both measures rank Player $1$ above Player $2$. To summarise, the two players have the same performance in the additional race but their relative ranking is exchanged by the rectangle index. When Player $1$ has two outstanding results of five points, this seems to be the rule rather than an exception as in the first case.
\end{example}

\begin{figure}[t!]
\centering

\begin{tikzpicture}
\begin{axis}[width=0.4\textwidth, 
height=0.35\textwidth,
tick label style={/pgf/number format/fixed},
symbolic x coords={($T1$)},
xtick = \empty,
xlabel = Player 1,
xlabel style = {font=\small},
ybar = 0.25cm,
ymax = 5.5,
ymin = 0,
ymajorgrids = true,
bar width = 1cm,
]

\addplot [red, pattern color = red, pattern = horizontal lines, very thick] coordinates{
(($T1$),5)
};

\addplot [red, pattern color = red, pattern = horizontal lines, very thick] coordinates{
(($T1$),2)
};
\end{axis}
\end{tikzpicture}
\hspace{0.1\textwidth}
\begin{tikzpicture}
\begin{axis}[width=0.4\textwidth, 
height=0.35\textwidth,
tick label style={/pgf/number format/fixed},
symbolic x coords={($T1$)},
xtick = \empty,
xlabel = Player 2,
xlabel style = {font=\small},
ybar = 0.25cm,
ymax = 5.5,
ymin = 0,
ymajorgrids = true,
bar width = 1cm,
]

\addplot [blue, pattern color = blue, pattern = dots, very thick] coordinates{
(($T1$),3)
};

\addplot [blue, pattern color = blue, pattern = dots, very thick] coordinates{
(($T1$),2)
};

\addplot [blue, pattern color = blue, pattern = dots, very thick] coordinates{
(($T1$),2)
};
\end{axis}
\end{tikzpicture}

\captionsetup{justification=centering}
\caption{The rectangle index does not satisfy depth relevance}
\label{Fig2}

\end{figure}


\begin{example} \label{Examp22}
\emph{The rectangle index does not satisfy depth relevance} \\
Take the score vectors $\mathbf{x} = \left[ 5, 2 \right]$ and $\mathbf{y} = \left[ 3, 2, 2 \right]$ of Players $1$ and $2$, plotted in Figure~\ref{Fig2}. Since $R(\mathbf{x}) = 1 \times 5 = 5$ and $R(\mathbf{y}) = 3 \times 2 = 6$ but both players have scored the same number of points in total (7), the rectangle index violates depth relevance.
The rectangle index prefers better core performance after a race is split into two.
On the other hand, $E(\mathbf{x}) = \sqrt{29} \approx 5.39$ and $E(\mathbf{y}) = \sqrt{17} \approx 4.12$, thus, the Euclidean index favours player $1$.
\end{example}

\begin{figure}[t!]
\centering

\begin{subfigure}{\textwidth}
\centering
\caption{The score vectors of Player $1$}
\label{Fig3a}

\begin{tikzpicture}
\begin{axis}[width=0.3\textwidth, 
height=0.35\textwidth,
tick label style={/pgf/number format/fixed},
symbolic x coords={($T1$)},
xtick = \empty,
xlabel = Original,
xlabel style = {font=\small},
ybar = 0.2cm,
ymax = 7.5,
ymin = 0,
ymajorgrids = true,
bar width = 0.8cm,
]

\addplot [red, pattern color = red, pattern = horizontal lines, very thick] coordinates{
(($T1$),3)
};

\addplot [red, pattern color = red, pattern = horizontal lines, very thick] coordinates{
(($T1$),3)
};
\end{axis}
\end{tikzpicture}
\hspace{0.1\textwidth}
\begin{tikzpicture}
\begin{axis}[width=0.3\textwidth, 
height=0.35\textwidth,
tick label style={/pgf/number format/fixed},
symbolic x coords={($T1$)},
xtick = \empty,
xlabel = One point is added,
xlabel style = {font=\small},
ybar = 0.2cm,
ymax = 7.5,
ymin = 0,
ymajorgrids = true,
bar width = 0.8cm,
]

\addplot [red, pattern color = red, pattern = horizontal lines, very thick] coordinates{
(($T1$),4)
};

\addplot [red, pattern color = red, pattern = horizontal lines, very thick] coordinates{
(($T1$),3)
};
\end{axis}
\end{tikzpicture}
\hspace{0.1\textwidth}
\begin{tikzpicture}
\begin{axis}[width=0.3\textwidth, 
height=0.35\textwidth,
tick label style={/pgf/number format/fixed},
symbolic x coords={($T1$)},
xtick = \empty,
xlabel = Four points are added,
xlabel style = {font=\small},
ybar = 0.2cm,
ymax = 7.5,
ymin = 0,
ymajorgrids = true,
bar width = 0.8cm,
]

\addplot [red, pattern color = red, pattern = horizontal lines, very thick] coordinates{
(($T1$),7)
};

\addplot [red, pattern color = red, pattern = horizontal lines, very thick] coordinates{
(($T1$),3)
};
\end{axis}
\end{tikzpicture}
\end{subfigure}

\vspace{0.2cm}
\begin{subfigure}{\textwidth}
\centering
\caption{The score vectors of Player $2$}
\label{Fig3b}

\begin{tikzpicture}
\begin{axis}[width=0.3\textwidth, 
height=0.35\textwidth,
tick label style={/pgf/number format/fixed},
symbolic x coords={($T1$)},
xtick = \empty,
xlabel = Original,
xlabel style = {font=\small},
ybar = 0.15cm,
ymax = 7.5,
ymin = 0,
ymajorgrids = true,
bar width = 0.6cm,
]

\addplot [blue, pattern color = blue, pattern = dots, very thick] coordinates{
(($T1$),2)
};

\addplot [blue, pattern color = blue, pattern = dots, very thick] coordinates{
(($T1$),2)
};

\addplot [blue, pattern color = blue, pattern = dots, very thick] coordinates{
(($T1$),2)
};
\end{axis}
\end{tikzpicture}
\hspace{0.1\textwidth}
\begin{tikzpicture}
\begin{axis}[width=0.3\textwidth, 
height=0.35\textwidth,
tick label style={/pgf/number format/fixed},
symbolic x coords={($T1$)},
xtick = \empty,
xlabel = One point is added,
xlabel style = {font=\small},
ybar = 0.15cm,
ymax = 7.5,
ymin = 0,
ymajorgrids = true,
bar width = 0.6cm,
]

\addplot [blue, pattern color = blue, pattern = dots, very thick] coordinates{
(($T1$),3)
};

\addplot [blue, pattern color = blue, pattern = dots, very thick] coordinates{
(($T1$),2)
};

\addplot [blue, pattern color = blue, pattern = dots, very thick] coordinates{
(($T1$),2)
};
\end{axis}
\end{tikzpicture}
\hspace{0.1\textwidth}
\begin{tikzpicture}
\begin{axis}[width=0.3\textwidth, 
height=0.35\textwidth,
tick label style={/pgf/number format/fixed},
symbolic x coords={($T1$)},
xtick = \empty,
xlabel = Four points are added,
xlabel style = {font=\small},
ybar = 0.15cm,
ymax = 7.5,
ymin = 0,
ymajorgrids = true,
bar width = 0.6cm,
]

\addplot [blue, pattern color = blue, pattern = dots, very thick] coordinates{
(($T1$),6)
};

\addplot [blue, pattern color = blue, pattern = dots, very thick] coordinates{
(($T1$),2)
};

\addplot [blue, pattern color = blue, pattern = dots, very thick] coordinates{
(($T1$),2)
};
\end{axis}
\end{tikzpicture}
\end{subfigure}

\captionsetup{justification=centering}
\caption{The rectangle index does not satisfy directional consistency}
\label{Fig3}

\end{figure}


\begin{example} \label{Examp23}
\emph{The rectangle index does not satisfy directional consistency} \\
Take the score vectors $\mathbf{x} = \left[ 3, 3, 0 \right]$ and $\mathbf{y} = \left[ 2, 2, 2 \right]$ of Players $1$ and $2$, plotted in Figures~\ref{Fig3a} and \ref{Fig3b}, respectively. Add the score vector $\mathbf{d} = \left[ 1, 0, 0 \right]$ to them. Since $R(\mathbf{x}) = 2 \times 3 = 3 \times 2 = R(\mathbf{y})$ and $R(\mathbf{x} + \mathbf{d}) = 2 \times 3 = 3 \times 2 = R(\mathbf{y} + \mathbf{d})$, directional consistency requires that $R(\mathbf{x} + 4\mathbf{d})$ equals $R(\mathbf{y} + 4\mathbf{d})$. But $R(\mathbf{x} + 4\mathbf{d}) = 1 \times 7 = 7$ and $R(\mathbf{y} + 4\mathbf{d}) = 1 \times 6 = 6$.
The rectangle index disregards further results if the performance in one race becomes outstanding.
On the other hand,  the Euclidean index always prefers Player $1$: $E(\mathbf{x}) = \sqrt{18} > E(\mathbf{y}) = \sqrt{12}$, $E(\mathbf{x} + \mathbf{d}) = \sqrt{25} > E(\mathbf{y} + \mathbf{d}) = \sqrt{17}$, and $E(\mathbf{x} + 4\mathbf{d}) = \sqrt{58} > E(\mathbf{y} + 4\mathbf{d}) = \sqrt{44}$.
\end{example}

\begin{table}[t]
  \centering
  \captionsetup{justification = centering}
  \caption{Axiomatic comparison of two bibliometric indices}
  \label{Table1}
\begin{threeparttable}
\rowcolors{1}{}{gray!20} 
    \begin{tabularx}{0.8\textwidth}{lCC} \toprule
    Property & Euclidean index & Rectangle index \\ \bottomrule
    Monotonicity & \textcolor{ForestGreen}{\ding{52}}$^\ast$ & \textcolor{ForestGreen}{\ding{52}}$^\ast$ \\
    Independence & \textcolor{ForestGreen}{\ding{52}}$^\ast$ & \textcolor{BrickRed}{\ding{55}} \\
    Depth relevance & \textcolor{ForestGreen}{\ding{52}}$^\ast$ & \textcolor{BrickRed}{\ding{55}} \\
    Scale invariance & \textcolor{ForestGreen}{\ding{52}}$^\ast$ & \textcolor{ForestGreen}{\ding{52}} \\
    Directional consistency & \textcolor{ForestGreen}{\ding{52}}$^\ast$ & \textcolor{BrickRed}{\ding{55}} \\
    Uniform citation & (\textcolor{BrickRed}{\ding{55}})\tnote{1} & \textcolor{ForestGreen}{\ding{52}}$^\ast$ \\
    Uniform equivalence & \textcolor{BrickRed}{\ding{55}} & \textcolor{ForestGreen}{\ding{52}}$^\ast$ \\ \toprule
    \end{tabularx}
\begin{tablenotes} \footnotesize
\item[1] The Euclidean index violates uniform citation but it will satisfy it if multiplied by $\sqrt{n}$.
\item
The combination of the axioms with an asterisk characterises the corresponding index.
\end{tablenotes}
\end{threeparttable}
\end{table}

Table~\ref{Table1} summarises the axiomatic comparison of the two indices. They grab a different aspect of the contestant's performance. The Euclidean index takes into account all points scored but rewards quality over quantity by summing up their squares. Contrarily, the rectangle index supports balanced achievements by ignoring unexpected peaks which are often due to mere luck or unusually favourable circumstances in sports.

The quantified achievements of the contestants can be used for at least two purposes: to measure their performance and to determine competitive balance in the championship. The first is revealed by the share of the total reward earned by a competitor. Competitive balance can be determined by evaluating the inequality in the achievements of the contestants, for example, according to the Herfindahl--Hirschman index or the Gini coefficient.

We do not say that any of the two approaches dominates the other. The choice should depend on the preferences of the decision-maker. However, if both measures indicate the same tendencies, this would be powerful evidence for a change in the corresponding direction. On the other hand, if the implications from the two underlying indices seem to contradict each other, one cannot be sure what happens in the competition.

\section{Data and implementation} \label{Sec3}

The UEFA Champions League, or simply the Champions League, is the most prestigious annual club football competition in Europe. The tournament format is fixed since the 2003/04 season (but will change from 2024/25 \citep{UEFA2022a}): a round-robin group stage with eight groups of four---altogether 32---teams each such that the top two teams from each group qualify to the Round of 16 in the knockout phase. 

The level of competition at the group stage is not very interesting because the Champions League slots are allocated in a deterministic way based on the ranking of UEFA member associations. For example, it is guaranteed that four clubs from the top four leagues of England, Germany, Italy, and Spain participate in the 2022/23 season. \citet{Csato2022b} provides more details on the Champions League qualification.

Tables~\ref{Table_A1} and \ref{Table_A2} in the Appendix show all participants in the knockout stage of the Champions League over the 20 seasons played between 2003 and 2023. In each year, there is exactly one winner (W), while one team is eliminated in the final (F), two are eliminated in the semifinals (SF), four in the quarterfinals (QF), and eight in the Round of 16 (R16).

As all teams have a national affiliation and the qualification depends on the ranking of UEFA member countries, at least three types of competitive balance can be distinguished:
\begin{itemize}
\item
the overall competitive balance at the level of clubs;
\item
the between-country competitive balance when only the national association of the teams is taken into consideration; and
\item
the within-country competitive balance, that is, how stable is the set of teams coming from a given national association.
\end{itemize}
Analogously, the performances of clubs and countries are worth measuring separately.

\begin{table}[t]
  \centering
  \caption{The weights assigned to the positions}
  \label{Table2}
  \rowcolors{1}{gray!20}{} 
    \begin{tabularx}{0.6\textwidth}{CCCCC} \toprule \hiderowcolors
    \multirow{2}[0]{*}{Position} & \multicolumn{4}{c}{Weighting method} \\
          & $W1$ & $W2$ & $W3$ & $W4$ \\ \bottomrule \showrowcolors
    W     & 16    & 5     & 6     & 1 \\
    F     & 8     & 4     & 5     & 1 \\
    SF    & 4     & 3     & 4     & 1 \\
    QF    & 2     & 2     & 3     & 1 \\
    R16   & 1     & 1     & 2     & 1 \\ \toprule
    \end{tabularx}
\end{table}

They all can be quantified by the methodology proposed in Section~\ref{Sec2}. In particular, any season is regarded as a race, and weights are associated with each position (W, F, SF, QF, R16) in each season according to Table~\ref{Table2}.
Weighting method $W1$ strongly favours the top positions, but note that there are 16 teams in the Round of 16 compared to the unique winner. Weighting method $W4$ is extremal by counting only whether a team has qualified for the knockout stage. Weightings $W2$ and $W3$ represent a middle course between the other two. \citet{Milanovic2005} uses a really sharp assignment of points to calculate Gini coefficients in the Champions League: 10 points for the winner, 6 points for the finalist, 2 points for the two semi-finalists, and 1 point for the four quarterfinalists.

Admittedly, our choice for the weights seems to be somewhat arbitrary. However, studying different sets of values and two bibliometric indices can indicate if the qualitative results are sensitive to this setting.

Both performance and competitive balance are assessed across subsequent seasons, hence an appropriate time window should be considered. As has been mentioned in the Introduction, a rolling five-year window can be justified by the seeding policy of the UEFA.

The positions of team $i$ (or country $i$) in a given period are collected into the score vector $\mathbf{x}^{(i)}$ that is transformed into a real number $f \left( \mathbf{x}^{(i)} \right)$ by the Euclidean or rectangle index. After that, the ``market share'' of club $i$ ($1 \leq i \leq n$) is calculated to directly obtain its performance:
\begin{equation} \label{eq1}
s_i = \frac{f \left( \mathbf{x}^{(i)} \right)}{\sum_{i=1}^n f \left( \mathbf{x}^{(i)} \right)},
\end{equation}

For competitive balance, the Herfindahl--Hirschman index ($\mathit{HHI}$) is computed as follows:
\begin{equation} \label{eq2}
\mathit{HHI} = \sum_{i=1}^n s_i^2 = \sum_{i=1}^n \left( \frac{f \left( \mathbf{x}^{(i)} \right)}{\sum_{i=1}^n f \left( \mathbf{x}^{(i)} \right)} \right)^2.
\end{equation}
A higher value of the $\mathit{HHI}$ indicates a \emph{lower} level of competitive balance.

Note that the minimum concentration is not zero in this hierarchical system since, e.g.\ there should be 16 teams in the Round of 16 in every year. However, this is not a problem because we will focus on relative changes.

\begin{example} \label{Examp31}
There are three Portuguese teams in our sample, Benfica, Porto, and Sporting CP. Let us concentrate on the five seasons from 2014/15 to 2018/19 and take the weighting scheme $W2$. The score vector of Benfica is $\left[ 2,1,0,0,0 \right]$, the score vector of Porto is $\left[ 2,2,1,1,0 \right]$, and the score vector of Sporting CP is $\left[ 0,0,0,0,0 \right]$. Consequently, the Euclidean index of Benfica is $E_B = \sqrt{5} \approx 2.24$, the rectangle index of Benfica is $R_B = 1 \times 2 = 2$, the Euclidean index of Porto is $E_P = \sqrt{10} \approx 3.16$, the rectangle index of Porto is $R_B = 2 \times 2 = 4$, while both measures are zero for Sporting CP.

The within-country competitive balance is
\[
\left( \frac{E_B}{E_B + E_P} \right)^2 + \left( \frac{E_P}{E_B + E_P} \right)^2 = \frac{5 + 10}{5 + 2 \sqrt{5} \sqrt{10} + 10} \approx 0.515
\]
under the Euclidean index, and it is
\[
\left( \frac{R_B}{R_B + R_P} \right)^2 + \left( \frac{R_P}{R_B + R_P} \right)^2 = \frac{4 + 16}{36} \approx 0.556
\]
under the rectangle index.
\end{example}

The number of seasons and clubs considered determines the maximal possible number of ``publications'' to account for. Therefore, in the case of the rectangle index, it is advised to use weighting methods $W2$, $W3$, and $W4$ when there are only a few opportunities to score points, thus, there is no need to provide many points for the best positions. However, if the number of qualifications to the knockout stage can be high, the weighting methods $W3$ and $W4$ usually result in the same ranking of the clubs (countries) by simply counting the number of times the Round of 16 is reached.

\begin{table}[t]
  \centering
  \caption{The total achievement of Porto over the 16 seasons}
  \label{Table3}
  \rowcolors{1}{gray!20}{} 
    \begin{tabularx}{\textwidth}{c CCCC} \toprule \hiderowcolors
    \multirow{2}[0]{*}{Position} & \multicolumn{4}{c}{Weighting method} \\
          & $W1$ & $W2$ & $W3$ & $W4$ \\ \bottomrule \showrowcolors
    W     & 16    & 5     & 6     & 1 \\
    QF    & 4     & 2     & 3     & 1 \\
    QF    & 4     & 2     & 3     & 1 \\
    QF    & 4     & 2     & 3     & 1 \\
    QF    & 4     & 2     & 3     & 1 \\
    R16   & 1     & 1     & 2     & 1 \\
    R16   & 1     & 1     & 2     & 1 \\
    R16   & 1     & 1     & 2     & 1 \\
    R16   & 1     & 1     & 2     & 1 \\
    R16   & 1     & 1     & 2     & 1 \\
    R16   & 1     & 1     & 2     & 1 \\
    R16   & 1     & 1     & 2     & 1 \\ 
    R16   & 1     & 1     & 2     & 1 \\ \toprule
    Rectangle index & $4 \times 5 = 20$    & $1 \times 13 = 13$    & $2 \times 13 = 26$    & $1 \times 13 = 13$ \\ \bottomrule
    \end{tabularx}
\end{table}

\begin{example} \label{Examp32}
Table~\ref{Table3} summarises the positions of Porto in the 20 Champions League seasons played between 2003 and 2023. Its achievement remains the same under weightings $W2$, $W3$, and $W4$ because the rectangle index is based on the 13 occasions when the club qualified to the Round of 16: the ``natural'' position of Porto is in the Round of 16. On the other hand, if weighting $W1$ is applied, the five participations in the quarterfinals will be counted.
\end{example}

Similarly, it would be misleading to use weighting method $W1$ with the Euclidean index because this implies that almost everything is determined by the top positions due to the squares in its mathematical formula.

\begin{example} \label{Examp33}
Consider Table~\ref{Table3}. The Euclidean index of Porto under weighting $W1$ is $\sqrt{16^2 + 4 \times 4^2 + 8 \times 1^2}$, that is, the influence of its single winning is 256 times larger than a qualification to the Round of 16. This seems to be unjustified.
\end{example}

\section{Results} \label{Sec4}

In the following, our empirical findings will be discussed. Section~\ref{Sec41} overviews some papers that are connected to measuring competitive balance in European football, which is followed by the study of the club and country performances in Section~\ref{Sec42} and the analysis of competitive balance in Section~\ref{Sec43}.  

\subsection{Related literature} \label{Sec41}

Several papers have examined competitive balance in European football leagues \citep{Koning2000, Szymanski2001, Goossens2006}. \citet{PawlowskiBreuerHovemann2010} investigate how the competitive situation in a domestic league is influenced by participation in the UEFA Champions League. \citet{HaanKoningvanWitteloostuijn2012} study in a theoretical model how the growing importance of the Champions League and the increased international trade in talent affect competitive balance within national competitions and the quality differences between them. Their results reveal that the Champions League in itself reduces inequality, both nationally and internationally, except for the very small countries.

There exists more limited research on competitive balance in European competitions. Regarding the UEFA Champions League, \citet{Milanovic2005} finds gradual deconcentration between 1963 and 1987, followed by a sharp reversal of this trend between 1988 and 2002.
According to \citet{Koning2009}, success in the Champions League has become more persistent between the time windows 1980--1998 and 1999--2007, mainly due to allowing multiple teams per country to play. The author argues that competitive balance in European tournaments is better to analyse at the level of national associations because most fans will identify with a domestic club from their home country, even if it is not their favourite.

\citet{PlumleyFlint2015} and \citet{Triguero-RuizAvila-Cano2023} study competitive balance in the group stage of the Champions League.
\citet{SchokkaertSwinnen2016} confirm the result of \citet{Milanovic2005}: the same teams are more likely to qualify for the knockout stage of the Champions League compared to its predecessor European Cup prior to the 1992/93 season. However, the uncertainty of which club wins beyond this stage has increased in the competition.
\citet{Bullough2018} analyses representation, performance, and revenue distribution in the Champions League from the 2003/04 to the 2016/17 seasons and finds that the current structure significantly benefits a small proportion of clubs.

\citet{RamchandaniPlumleyMondalMillarWilson2023} examine the competitive balance of the UEFA Champions League before and after the implementation of the UEFA Financial Fair Play (FFP) regulations in 2011. Regarding the knockout stage, they simply calculate the number of unique clubs that have progressed to the knockout stage, as well as the number of occasions that clubs from a particular country played in the semifinals and the final, or won the tournament. Their analysis shows a decline in competitive balance in the post-FFP era. Even though it is also stated that the knockout stage has been dominated by clubs from the top five European leagues since 1992/93, this is not quantified by a particular metric.


Competitive balance may depend on many factors, including the score system \citep{Haugen2008}, the scoring rate \citep{ScarfParmaMcHale2019}, major rule changes \citep{KentCaudillMixon2013}, participation in other competitions \citep{Moffat2020}, the tournament format \citep{Csato2020b, Csato2021b, LasekGagolewski2018, McGarrySchutz1997, ScarfYusofBilbao2009}, the seeding procedure \citep{CoronaForrestTenaWiper2019, DagaevRudyak2019}, the number of competitors \citep{Cairns1987}, as well as promotion and relegation \citep{Noll2002, BuzzacchiSzymanskiValletti2003}. The latter two, to some extent, are analogous to the changing set of Champions League participants between two subsequent seasons. Although these issues will not be discussed here, the consideration of the causes that govern the dynamics of competitive balance in the UEFA Champions League can be a promising direction for future research.

\subsection{The performances of clubs and national associations} \label{Sec42}

\begin{table}[t]
  \centering
  \caption{Club performances between 2003 and 2023: the top five}
  \label{Table4}
\begin{threeparttable}
  \rowcolors{1}{gray!20}{} 
    \begin{tabularx}{0.8\textwidth}{LLL} \toprule \hiderowcolors
    \multicolumn{3}{c}{Euclidean index} \\
    $W2$    & $W3$    & $W4$ \\ \bottomrule \showrowcolors    
    Real Madrid & Real Madrid & Real Madrid \\
    Barcelona & Barcelona & Bayern Munich \\
    Bayern Munich & Bayern Munich & Barcelona\tnote{1} \\
    Chelsea & Chelsea & Chelsea\tnote{1} \\
    Liverpool & Liverpool & Arsenal \\ \toprule
    \end{tabularx}
\vspace{0.25cm}
    \begin{tabularx}{0.8\textwidth}{LLL} \toprule \hiderowcolors
    \multicolumn{3}{c}{Rectangle index} \\
    $W1$    & $W2$    & $W3$ \\ \bottomrule \showrowcolors    
    Real Madrid & Real Madrid & Bayern Munich \\
    Barcelona & Bayern Munich & Real Madrid \\
    Liverpool & Barcelona & Barcelona \\
    Bayern Munich\tnote{2} & Chelsea & Chelsea \\
    Chelsea\tnote{2} & Liverpool & Arsenal \\ \toprule
    \end{tabularx}
\begin{tablenotes} \footnotesize
\item[1] Barcelona and Chelsea are tied for third place.
\item[2] Bayern Munich and Chelsea are tied for fourth place.
\end{tablenotes}
\end{threeparttable}
\end{table}

\noindent
First, we focus on the achievements of teams and countries.
Table~\ref{Table4} shows the five highest-ranked clubs in the whole period from 2003/04 to 2022/23 if their performances are quantified by the two indices and the three weighting methods suggested. The top of the ranking is remarkably robust: the first three positions are almost always occupied by the Spanish clubs Barcelona and Real Madrid together with the German Bayern Munich.
Real Madrid turns out to be the best team except for the rectangle index and weighting $W3$.
They are followed by three English clubs: Arsenal, Chelsea, and Liverpool. While Arsenal often qualified for the Round of 16, its performance there remained modest, thus, it is favoured by a less sharp weighting. For instance, under the weighting scheme $W2$, Arsenal is only the $13$th best team by the Euclidean index. Its achievement stands in stark contrast to the record of Liverpool, which was missing from the knockout stage of the Champions League in several seasons but achieved good results when qualifying for this phase.

\input{Figure4_club_performances}

Figure~\ref{Fig4} compares the achievements of these six teams over the last five seasons in the corresponding year. According to all measures, the performance of Arsenal has gradually declined. Chelsea shows a similar decrease except for the last five-year periods, while Real Madrid has been at its peak in 2018 since it has won in three subsequent seasons. However, note that formula~\eqref{eq1} also depends on the achievements of the other teams.
Barcelona (Bayern Munich) has been the best towards the beginning (middle) of our time window, and Liverpool was absent from the knockout stage of the UEFA Champions League between the seasons 2009/10 and 2016/17. Nonetheless, in this set of six clubs, it cannot be said that any team unambiguously dominated another one during the whole period considered.

\begin{table}[t]
  \centering
  \caption{Country performances between 2003 and 2023: the top five}
  \label{Table5}
\centerline{
  \rowcolors{1}{gray!20}{} 
    \begin{tabularx}{1\textwidth}{LLL LLL} \toprule \hiderowcolors
    \multicolumn{3}{c}{Euclidean index} & \multicolumn{3}{c}{Rectangle index} \\
    $W2$    & $W3$    & $W4$    & $W1$    & $W2$    & $W3$ \\ \bottomrule \showrowcolors
    England & England & England & Spain & England & England \\
    Spain & Spain & Spain & England & Spain & Spain \\
    Germany & Germany & Italy & Italy & Italy & Italy \\
    Italy & Italy & Germany & Germany & Germany & Germany \\
    France & France & France & France & France & France \\ \toprule
    \end{tabularx}
}
\end{table}

Figure~\ref{Fig4} can be somewhat misleading as the worsening performance of Arsenal and Chelsea mainly reflects the high level of competition in the English Premier League: they have failed to qualify for the knockout stage of the Champions League because they could not have played in the group stage.
Therefore, Table~\ref{Table5} turns to the comparison of national associations. The clubs of the five top leagues (England, France, Germany, Italy, Spain) clearly prevail in the Champions League. According to our measures, even their ranking is obvious except for the first place disputed by England and Spain, and the third position disputed by Germany and Italy.

\input{Figure5_country_performances}

However, Figure~\ref{Fig5} shows that clear trends can be observed in the achievements of the clubs from these leagues. In particular, the performance of English and Italian teams have decreased in the first half of the period, while German and Spanish clubs have become more competitive until the middle of the time window considered. However, the former dominance of Spanish clubs is threatened in recent years. In contrast to Figure~\ref{Fig4}, the relative order of certain national associations seems to be robust: Spain is consistently better than Germany, England is consistently better than Italy, and France has the weakest league among the top five. Spain has probably outperformed Italy, too.

To conclude, the comparison of countries is easier than the comparison of individual clubs. While this is almost a straightforward implication since the fluctuations in the performances of teams are smoothed out at the level of national associations, it yields at least two important lessons. First, UEFA follows a reasonable policy by founding the Champions League qualification on the ranking of associations rather than the ranking of individual teams \citep{Csato2022b}. Second, the current use of UEFA club coefficients for seeding in the UEFA Champions League \citep{Csato2020a} and the second-tier competition UEFA Europa League is perhaps flawed; it would be better to draw the groups on the basis of country characteristics. For instance, \citet{Guyon2015b} recommends labelling all clubs by their finishing positions in their domestic leagues and quantifying their performances in this way instead of focusing separately on each club. This approach could be fairer for a team emerging from nil: even though it was one of the greatest sporting stories of all time when Leicester City became the champion in the English Premier League despite overwhelming odds \citep{BBC2016}, the team had the second smallest UEFA club coefficient in the next season of the Champions League.

\subsection{Competitive balance} \label{Sec43}

\begin{figure}[t!]
\centering

\pgfkeys{/pgf/number format/.cd,1000 sep={}}
\begin{tikzpicture}
\begin{axis}[
name = axis1,
width = 1\textwidth, 
height = 0.6\textwidth,
xmin = 2008,
xmax = 2023,
ymajorgrids,
scaled ticks = false,
y tick label style = {/pgf/number format/.cd,fixed,precision=3},
y label style={at={(current axis.above origin)},rotate=270,anchor=south east,font=\small},
ylabel = {$\mathit{HHI}$},
legend entries={Euclidean $W2 \qquad$,Euclidean $W3 \qquad$,Euclidean $W4$,Rectangle $W1 \qquad$,Rectangle $W2 \qquad$,Rectangle $W3$},
legend style = {at={(0.5,-0.1)},anchor=north,legend columns = 3,font=\small}
]
\addplot[blue,thick,dashed,mark=o,mark options={solid}] coordinates {
(2008,0.0432132294126968)
(2009,0.0501074794919575)
(2010,0.0461559168626377)
(2011,0.0462614697465505)
(2012,0.0437951169521821)
(2013,0.0395589531316952)
(2014,0.0412710354713076)
(2015,0.0446548468099456)
(2016,0.045132673000189)
(2017,0.0497748536168644)
(2018,0.0485244408169867)
(2019,0.0468632630927995)
(2020,0.0422268584580961)
(2021,0.0440384657961595)
(2022,0.0448730619135468)
(2023,0.0430305122334808)
};
\addplot[black,thick,dotted,mark=square,mark options={solid}] coordinates {
(2008,0.0381310067930336)
(2009,0.0430938963788451)
(2010,0.0400977242031707)
(2011,0.0403384660246522)
(2012,0.0371807471954436)
(2013,0.0338121677173954)
(2014,0.0354584480110615)
(2015,0.0381687274210453)
(2016,0.0377510995966801)
(2017,0.0415702001447882)
(2018,0.0408107750826244)
(2019,0.0403883359276049)
(2020,0.0370153798705695)
(2021,0.0395526665235893)
(2022,0.0395440674722554)
(2023,0.0384203267120476)
};
\addplot[ForestGreen,thick,,mark=otimes,mark options={solid}] coordinates {
(2008,0.0329134433324317)
(2009,0.0360752376067066)
(2010,0.0338231965071848)
(2011,0.0340723536743908)
(2012,0.030320895856151)
(2013,0.0278709884536849)
(2014,0.0296623678105144)
(2015,0.0317950694051789)
(2016,0.0304153636080284)
(2017,0.0333085213447434)
(2018,0.0331779016273282)
(2019,0.0340285854361346)
(2020,0.0316269807498857)
(2021,0.0345811005252203)
(2022,0.0335520029402893)
(2023,0.0333085213447434)
};
\addplot[red,thick,dashed,mark=diamond,mark options={solid}] coordinates {
(2008,0.0618311533888228)
(2009,0.0837964788984904)
(2010,0.0797941248555723)
(2011,0.0839510204081633)
(2012,0.084452479338843)
(2013,0.0780977437158626)
(2014,0.0654549099535568)
(2015,0.0772244897959184)
(2016,0.0780246913580247)
(2017,0.108742905446202)
(2018,0.1347)
(2019,0.103403944485026)
(2020,0.0989043097151206)
(2021,0.0760132313141101)
(2022,0.0815538319129241)
(2023,0.0613774104683196)
};
\addplot[brown,thick,dotted,mark=asterisk,mark options={solid}] coordinates {
(2008,0.0474103294616115)
(2009,0.0537201953461649)
(2010,0.0503231763619575)
(2011,0.0521698984302862)
(2012,0.0535979097521362)
(2013,0.0504746943514788)
(2014,0.0495833333333333)
(2015,0.0545833333333333)
(2016,0.056512)
(2017,0.0582610477426405)
(2018,0.063114134542706)
(2019,0.0541628416072311)
(2020,0.0519773239532819)
(2021,0.0484538592646701)
(2022,0.0520792544443575)
(2023,0.046830614398182)
};
\addplot[orange,thick,,mark=triangle,mark options={solid}] coordinates {
(2008,0.0430290543163427)
(2009,0.0475967859751643)
(2010,0.0448960302457467)
(2011,0.0460121822185364)
(2012,0.044921875)
(2013,0.0419921875)
(2014,0.04212193615936)
(2015,0.0455226824457594)
(2016,0.04725)
(2017,0.048836291913215)
(2018,0.0505050505050505)
(2019,0.0455305501493928)
(2020,0.0441051506263535)
(2021,0.0441795918367347)
(2022,0.0460800403989395)
(2023,0.0429953078617256)
};
\end{axis}
\end{tikzpicture}

\captionsetup{justification=centerfirst}
\caption{Competition among the clubs in five-year periods between 2003 and 2023 \\
\footnotesize{The year on the x-axis indicates the finishing year of the last season in the period, e.g.\ 2018 corresponds to the seasons from 2013/14 to 2017/18.}}
\label{Fig6}

\end{figure}


Figure~\ref{Fig6} presents the dynamics of the Herfindahl--Hirschman index over five-year periods between 2003 and 2023 at the level of clubs. All six methods (three weighting schemes with two bibliometric indices) fluctuate without a clear trend.

\begin{figure}[t]
\centering

\pgfkeys{/pgf/number format/.cd,1000 sep={}}
\begin{tikzpicture}
\begin{axis}[
name = axis1,
width = 1\textwidth, 
height = 0.6\textwidth,
xmin = 2008,
xmax = 2023,
ymajorgrids,
scaled ticks = false,
tick label style = {/pgf/number format/fixed},
y tick label style = {/pgf/number format/.cd,fixed,precision=3},
y label style={at={(current axis.above origin)},rotate=270,anchor=south east,font=\small},
ylabel = {$\mathit{HHI}$},
legend entries={Euclidean $W2 \qquad$,Euclidean $W3 \qquad$,Euclidean $W4$,Rectangle $W1 \qquad$,Rectangle $W2 \qquad$,Rectangle $W3$},
legend style = {at={(0.5,-0.1)},anchor=north,legend columns = 3,font=\small}
]
\addplot[blue,thick,dashed,mark=o,mark options={solid}] coordinates {
(2008,0.132082369363195)
(2009,0.152428192517965)
(2010,0.13847869954545)
(2011,0.13024286687572)
(2012,0.123568014411118)
(2013,0.120762420927864)
(2014,0.122118531536249)
(2015,0.126431874103082)
(2016,0.123514840949684)
(2017,0.137262396690394)
(2018,0.143052958580698)
(2019,0.138599049639233)
(2020,0.138277976248122)
(2021,0.150118934284679)
(2022,0.14800703993551)
(2023,0.155193846562522)
};
\addplot[black,thick,dotted,mark=square,mark options={solid}] coordinates {
(2008,0.120351585099252)
(2009,0.136637482505723)
(2010,0.125765729837041)
(2011,0.117463778042631)
(2012,0.110227199796418)
(2013,0.108307219858401)
(2014,0.108659405425697)
(2015,0.110666891523614)
(2016,0.105932267764646)
(2017,0.117391895442121)
(2018,0.122015407507423)
(2019,0.123503598209656)
(2020,0.125766223170288)
(2021,0.141228246055353)
(2022,0.135575104848424)
(2023,0.146572622409621)
};
\addplot[ForestGreen,thick,,mark=otimes,mark options={solid}] coordinates {
(2008,0.104611515422675)
(2009,0.116767855257564)
(2010,0.109311673822514)
(2011,0.10139958996803)
(2012,0.0936587525838608)
(2013,0.0925421207704084)
(2014,0.0911920294441918)
(2015,0.0910966831249511)
(2016,0.0849155369546702)
(2017,0.0933983302850742)
(2018,0.097332470804878)
(2019,0.10436410418008)
(2020,0.108484946577742)
(2021,0.129323000451528)
(2022,0.118323593919019)
(2023,0.134344428140827)
};
\addplot[red,thick,dashed,mark=diamond,mark options={solid}] coordinates {
(2008,0.183662318147668)
(2009,0.226863905325444)
(2010,0.201218231463978)
(2011,0.194507318545345)
(2012,0.18698347107438)
(2013,0.174640172484121)
(2014,0.181640625)
(2015,0.209870419740839)
(2016,0.197314049586777)
(2017,0.249228395061728)
(2018,0.299606474907704)
(2019,0.255709325639396)
(2020,0.217309145880574)
(2021,0.202880859375)
(2022,0.21875696146135)
(2023,0.221759534668517)
};
\addplot[brown,thick,dotted,mark=asterisk,mark options={solid}] coordinates {
(2008,0.176206509539843)
(2009,0.198706009214783)
(2010,0.173407202216067)
(2011,0.1724)
(2012,0.160707420222991)
(2013,0.151455739633369)
(2014,0.156553279090285)
(2015,0.16320645905421)
(2016,0.153257386970452)
(2017,0.1704)
(2018,0.171)
(2019,0.164697542533081)
(2020,0.166284265185364)
(2021,0.169665499335829)
(2022,0.176512287334594)
(2023,0.175094517958412)
};
\addplot[orange,thick,,mark=triangle,mark options={solid}] coordinates {
(2008,0.150632866004518)
(2009,0.165008036660853)
(2010,0.154124149659864)
(2011,0.160032653061225)
(2012,0.146609213091208)
(2013,0.135726261240665)
(2014,0.138871956685432)
(2015,0.144215314113744)
(2016,0.136788284307537)
(2017,0.154024489795918)
(2018,0.157946888624653)
(2019,0.154075083366202)
(2020,0.160188995579942)
(2021,0.170858100899253)
(2022,0.1678125)
(2023,0.170556691485668)
};
\end{axis}
\end{tikzpicture}

\captionsetup{justification=centerfirst}
\caption{Competition among the countries in five-year periods between 2003 and 2023 \\
\footnotesize{The year on the x-axis indicates the finishing year of the last season in the period, e.g.\ 2018 corresponds to the seasons from 2013/14 to 2017/18.}}
\label{Fig7}

\end{figure}


Figure~\ref{Fig7} follows another approach by focusing on the competition between the countries. Except for the rectangle index under weighting $W1$, the other five measures uncover a deterioration of competitive balance at the level of national associations in the second half of the period.

\begin{figure}[t!]
\centering

\pgfkeys{/pgf/number format/.cd,1000 sep={}}
\begin{tikzpicture}
\begin{axis}[
name = axis1,
width = 1\textwidth, 
height = 0.6\textwidth,
xmin = 2008,
xmax = 2023,
ymajorgrids,
scaled ticks = false,
tick label style = {/pgf/number format/fixed},
y tick label style = {/pgf/number format/.cd,fixed,precision=3},
y label style={at={(current axis.above origin)},rotate=270,anchor=south east,font=\small},
ylabel = {$\mathit{HHI}$},
legend entries={Euclidean $W2 \qquad$,Euclidean $W3 \qquad$,Euclidean $W4$,Rectangle $W1 \qquad$,Rectangle $W2 \qquad$,Rectangle $W3$},
legend style = {at={(0.5,-0.1)},anchor=north,legend columns = 3,font=\small}
]
\addplot[blue,thick,dashed,mark=o,mark options={solid}] coordinates {
(2008,0.592265602817056)
(2009,0.639824790495624)
(2010,0.652091617595872)
(2011,0.658790963460826)
(2012,0.642761860653076)
(2013,0.642761860653076)
(2014,0.652091617595872)
(2015,0.652091617595872)
(2016,0.642761860653076)
(2017,0.673551751199573)
(2018,0.686057869698992)
(2019,0.648890822131322)
(2020,0.669673835345223)
(2021,0.686057869698992)
(2022,0.665926157982717)
(2023,0.658790963460826)
};
\addplot[black,thick,dotted,mark=square,mark options={solid}] coordinates {
(2008,0.58153747862377)
(2009,0.606731373234781)
(2010,0.610010509545971)
(2011,0.619339070160169)
(2012,0.60460641942092)
(2013,0.60460641942092)
(2014,0.610010509545971)
(2015,0.610010509545971)
(2016,0.597527792902829)
(2017,0.621815728188064)
(2018,0.628287566909727)
(2019,0.612260438104376)
(2020,0.633777255386176)
(2021,0.65913480545812)
(2022,0.638093895091543)
(2023,0.636635123117083)
};
\addplot[ForestGreen,thick,,mark=otimes,mark options={solid}] coordinates {
(2008,0.555555555555556)
(2009,0.555555555555556)
(2010,0.550008594033486)
(2011,0.561600640512513)
(2012,0.550008594033486)
(2013,0.550008594033486)
(2014,0.550008594033486)
(2015,0.550008594033486)
(2016,0.535898384862245)
(2017,0.550008594033486)
(2018,0.550008594033486)
(2019,0.555555555555556)
(2020,0.575434341467906)
(2021,0.61275965371367)
(2022,0.59215435776876)
(2023,0.601888306193515)
};
\addplot[red,thick,dashed,mark=diamond,mark options={solid}] coordinates {
(2008,0.702479338842975)
(2009,0.722222222222222)
(2010,0.710915081305133)
(2011,0.734072022160665)
(2012,0.710915081305133)
(2013,0.710915081305133)
(2014,0.710915081305133)
(2015,0.710915081305133)
(2016,0.68)
(2017,0.710915081305133)
(2018,0.710915081305133)
(2019,0.722222222222222)
(2020,0.759509769915597)
(2021,0.818204772124732)
(2022,0.787465281970776)
(2023,0.802469135802469)
};
\addplot[brown,thick,dotted,mark=asterisk,mark options={solid}] coordinates {
(2008,0.697133585722012)
(2009,0.697133585722012)
(2010,0.685559519091029)
(2011,0.709342560553633)
(2012,0.68)
(2013,0.68)
(2014,0.685559519091029)
(2015,0.685559519091029)
(2016,0.654320987654321)
(2017,0.696171959908224)
(2018,0.7012371372413)
(2019,0.702479338842975)
(2020,0.745805258290395)
(2021,0.806648279866454)
(2022,0.770068224706053)
(2023,0.78125)
};
\addplot[orange,thick,,mark=triangle,mark options={solid}] coordinates {
(2008,0.68)
(2009,0.68)
(2010,0.6653125)
(2011,0.6953125)
(2012,0.6653125)
(2013,0.6653125)
(2014,0.6653125)
(2015,0.6653125)
(2016,0.625)
(2017,0.6653125)
(2018,0.6653125)
(2019,0.68)
(2020,0.7278125)
(2021,0.8003125)
(2022,0.7628125)
(2023,0.78125)
};
\end{axis}
\end{tikzpicture}

\captionsetup{justification=centerfirst}
\caption{Competition in five-year periods between 2003 and 2023: the top five associations versus other countries \\
\footnotesize{The year on the x-axis indicates the finishing year of the last season in the period, e.g.\ 2018 corresponds to the seasons from 2013/14 to 2017/18.}}
\label{Fig8}

\end{figure}


Competitive balance can also be studied in the relation of big and small leagues. For instance, the 2019/20 season was the single one when only teams from the top five national associations qualified for the Round of 16.
Indeed, Figure~\ref{Fig8} shows some evidence that the uncertainty of outcome has recently declined from this perspective.

\input{Figure9_within_country_competitive_balance}

Finally, Figure~\ref{Fig9} analyses the third variant of competitive balance by focusing on the top five leagues, England (with 7 clubs in our dataset), France (6), Germany (10), Italy (8), and Spain (10). Here, there are some clear tendencies in the uncertainty of outcome. England and, recently, Italy have managed to improve competitive balance. The highest $\mathit{HHI}$s in France and Italy can be attributed to the dominance of one club, Juventus (the Italian champion between the seasons 2011/12 and 2019/20) and Paris Saint-Germain (the French champion since 2012/13 except for two seasons), respectively.
While these observations concern the competitiveness of domestic leagues, it is highlighted from an unusual perspective by focusing on the international performance of the leading clubs.

\section{Concluding remarks} \label{Sec5}

This paper has aimed to demonstrate how bibliometric indices can be used to measure the level of competition in a knockout tournament. Although our approach involves the choice of two crucial variables, the scientometric index and the weighting method, most results for the UEFA Champions League seem to be remarkably robust concerning these parameters. Basically, neither the performance of clubs and countries nor the three types of competitive balance depend on them to a great extent.

Nevertheless, our results can be sensitive to the luck of the draw in the knockout stage: as the weighting system awards more points to the clubs as they progress in the tournament, uncertainty of outcome can decrease merely because strong clubs are drawn early against each other. This, together with the low scoring nature of the sport, might make the final result of a season relatively noisy.
On the other hand, the public views only the outcome of the matches and is probably not satisfied if computer simulations do not show a deterioration in the level of competition but the same clubs almost always win in the real world.

The main findings of the work can be summarised as follows:
\begin{itemize}
\item
The performance of national associations varies less than the performance of individual clubs;
\item
In the middle of the period considered, Spanish teams have gained prominence at the expense of English and Italian teams, but this trend has reversed in the last years;
\item
Even though competitive balance has not changed at the level of clubs since 2003/04, there is some evidence for a recent decline at the level of countries, and, especially, in the relation of top five versus other leagues.
\end{itemize}
The relative stability of competitive balance in the last two decades supports the results of \citet{SchokkaertSwinnen2016} and can be explained by the qualification rules of the UEFA Champions League that allow multiple teams from the strongest leagues to compete. This movement partially contradicts the expected effects of free labour markets in European football, which should have increased the inequality of results among clubs \citep{Milanovic2005}.

The current study may have important policy lessons for sports administrators.
First, since the performances of national associations are more stable than the results of individual clubs, it would be better to build the seeding in the UEFA Champions League group stage upon association coefficients rather than club coefficients, accounting for the positions that the teams have achieved in their domestic leagues. This change would favour leagues where the competition within the country is stronger, and the set of teams qualifying for the Champions League is more volatile across the seasons.
The methodology proposed here can even be used to determine the ranking of national associations for the UEFA access list.

Second, there is no clear relationship between the competitive balance in a national league and the competitiveness of its clubs in the UEFA Champions League: the aggregated performances of English and Italian teams have both declined in the middle of the period considered (Figure~\ref{Fig5}), however, the competition within the countries has changed in the opposite way (Figure~\ref{Fig9}). A possible explanation for the case of Italy can be that even though domination in the domestic league helps a club to buy better players, the decreased incentives to be efficient at the home front may be detrimental to its international performance.

\section*{Acknowledgements}
\addcontentsline{toc}{section}{Acknowledgements}
\noindent
\emph{L\'aszl\'o Csat\'o}, the father of the first author has helped in carrying out the calculations. \\
Six anonymous reviewers provided valuable comments and suggestions on an earlier draft. \\
We are grateful to the \href{https://en.wikipedia.org/wiki/Wikipedia_community}{Wikipedia community} for collecting and structuring invaluable information on the sports tournaments discussed.

\bibliographystyle{apalike}
\bibliography{All_references}

\clearpage

\section*{Appendix}
\addcontentsline{toc}{section}{Appendix}

\renewcommand\thetable{A.\arabic{table}}
\setcounter{table}{0}

\makeatletter
\renewcommand\p@subtable{A.\arabic{table}}
\makeatother

\renewcommand\thefigure{A.\arabic{figure}}
\setcounter{figure}{0}

\makeatletter
\renewcommand\p@subfigure{A.\arabic{figure}}
\makeatother

\begin{sidewaystable}
  \centering
  \caption{Clubs in the knockout stage of the Champions League between the seasons 2003/04 and 2022/23 I.}
  \label{Table_A1}
\begin{adjustbox}{max width=\textwidth}
\begin{threeparttable}
    \rowcolors{1}{}{gray!20}
    \begin{tabular}{ll ccccc ccccc ccccc ccccc} \toprule
    Team  & Country & 2003  & 2004  & 2005  & 2006  & 2007  & 2008  & 2009  & 2010  & 2011  & 2012  & 2013  & 2014  & 2015  & 2016  & 2017  & 2018  & 2019  & 2020  & 2021  & 2022 \\ \bottomrule
    Ajax  & Netherlands & ---   & ---   & R16   & ---   & ---   & ---   & ---   & ---   & ---   & ---   & ---   & ---   & ---   & ---   & ---   & SF    & ---   & ---   & R16   & --- \\
    APOEL & Cyprus & ---   & ---   & ---   & ---   & ---   & ---   & ---   & ---   & QF    & ---   & ---   & ---   & ---   & ---   & ---   & ---   & ---   & ---   & ---   & --- \\
    Arsenal & England & QF    & R16   & F     & R16   & QF    & SF    & QF    & R16   & R16   & R16   & R16   & R16   & R16   & R16   & ---   & ---   & ---   & ---   & ---   & --- \\
    Atalanta & Italy & ---   & ---   & ---   & ---   & ---   & ---   & ---   & ---   & ---   & ---   & ---   & ---   & ---   & ---   & ---   & ---   & QF    & R16   & ---   & --- \\
    Atl\'etico Madrid & Spain & ---   & ---   & ---   & ---   & ---   & R16   & ---   & ---   & ---   & ---   & F     & QF    & F     & SF    & ---   & R16   & QF    & R16   & QF    & --- \\
    Barcelona & Spain & ---   & R16   & W     & R16   & SF    & W     & SF    & W     & SF    & SF    & QF    & W     & QF    & QF    & QF    & SF    & QF    & R16   & ---   & --- \\
    Basel & Switzerland & ---   & ---   & ---   & ---   & ---   & ---   & ---   & ---   & R16   & ---   & ---   & R16   & ---   & ---   & R16   & ---   & ---   & ---   & ---   & --- \\
    Bayer Leverkusen & Germany & ---   & R16   & ---   & ---   & ---   & ---   & ---   & ---   & R16   & ---   & R16   & R16   & ---   & R16   & ---   & ---   & ---   & ---   & ---   & --- \\
    Bayern Munich & Germany & R16   & QF    & R16   & QF    & ---   & QF    & F     & R16   & F     & W     & SF    & SF    & SF    & QF    & SF    & R16   & W     & QF    & QF    & QF \\
    Benfica & Portugal & ---   & ---   & QF    & ---   & ---   & ---   & ---   & ---   & QF    & ---   & ---   & ---   & QF    & R16   & ---   & ---   & ---   & ---   & QF    & QF \\
    Be{\c s}ikta{\c s} & Turkey & ---   & ---   & ---   & ---   & ---   & ---   & ---   & ---   & ---   & ---   & ---   & ---   & ---   & ---   & R16   & ---   & ---   & ---   & ---   & --- \\
    Bordeaux & France & ---   & ---   & ---   & ---   & ---   & ---   & QF    & ---   & ---   & ---   & ---   & ---   & ---   & ---   & ---   & ---   & ---   & ---   & ---   & --- \\
    Borussia Dortmund & Germany & ---   & ---   & ---   & ---   & ---   & ---   & ---   & ---   & ---   & F     & QF    & R16   & ---   & QF    & ---   & R16   & R16   & QF    & ---   & R16 \\
    Borussia M\"onchengladbach & Germany & ---   & ---   & ---   & ---   & ---   & ---   & ---   & ---   & ---   & ---   & ---   & ---   & ---   & ---   & ---   & ---   & ---   & R16   & ---   & --- \\
    Celta Vigo & Spain & R16   & ---   & ---   & ---   & ---   & ---   & ---   & ---   & ---   & ---   & ---   & ---   & ---   & ---   & ---   & ---   & ---   & ---   & ---   & --- \\
    Celtic & Scotland & ---   & ---   & ---   & R16   & R16   & ---   & ---   & ---   & ---   & R16   & ---   & ---   & ---   & ---   & ---   & ---   & ---   & ---   & ---   & --- \\
    Chelsea & England & SF    & SF    & R16   & SF    & F     & SF    & R16   & QF    & W     & ---   & SF    & R16   & R16   & ---   & R16   & ---   & R16   & W     & QF    & QF \\
    Club Brugge & Belgium & ---   & ---   & ---   & ---   & ---   & ---   & ---   & ---   & ---   & ---   & ---   & ---   & ---   & ---   & ---   & ---   & ---   & ---   & ---   & R16 \\
    Copenhagen & Denmark & ---   & ---   & ---   & ---   & ---   & ---   & ---   & R16   & ---   & ---   & ---   & ---   & ---   & ---   & ---   & ---   & ---   & ---   & ---   & --- \\
    CSKA Moscow & Russia & ---   & ---   & ---   & ---   & ---   & ---   & QF    & ---   & R16   & ---   & ---   & ---   & ---   & ---   & ---   & ---   & ---   & ---   & ---   & --- \\
    Deportivo La Coru{\~ n}a & Spain & SF    & ---   & ---   & ---   & ---   & ---   & ---   & ---   & ---   & ---   & ---   & ---   & ---   & ---   & ---   & ---   & ---   & ---   & ---   & --- \\
    Dinamo Kyiv & Ukraine & ---   & ---   & ---   & ---   & ---   & ---   & ---   & ---   & ---   & ---   & ---   & ---   & R16   & ---   & ---   & ---   & ---   & ---   & ---   & --- \\
    Eintracht Frankfurt & Germany & ---   & ---   & ---   & ---   & ---   & ---   & ---   & ---   & ---   & ---   & ---   & ---   & ---   & ---   & ---   & ---   & ---   & ---   & ---   & R16 \\
    Fenerbahce & Turkey & ---   & ---   & ---   & ---   & QF    & ---   & ---   & ---   & ---   & ---   & ---   & ---   & ---   & ---   & ---   & ---   & ---   & ---   & ---   & --- \\
    Fiorentina & Italy & ---   & ---   & ---   & ---   & ---   & ---   & R16   & ---   & ---   & ---   & ---   & ---   & ---   & ---   & ---   & ---   & ---   & ---   & ---   & --- \\
    Galatasaray & Turkey & ---   & ---   & ---   & ---   & ---   & ---   & ---   & ---   & ---   & QF    & R16   & ---   & ---   & ---   & ---   & ---   & ---   & ---   & ---   & --- \\
    Gent  & Belgium & ---   & ---   & ---   & ---   & ---   & ---   & ---   & ---   & ---   & ---   & ---   & ---   & R16   & ---   & ---   & ---   & ---   & ---   & ---   & --- \\
    Internazionale & Italy & ---   & QF    & QF    & R16   & R16   & R16   & W     & QF    & R16   & ---   & ---   & ---   & ---   & ---   & ---   & ---   & ---   & ---   & R16   & F \\
    Juventus & Italy & R16   & QF    & QF    & ---   & ---   & R16   & ---   & ---   & ---   & QF    & ---   & F     & R16   & F     & QF    & QF    & R16   & R16   & R16   & --- \\
    Lazio & Italy & ---   & ---   & ---   & ---   & ---   & ---   & ---   & ---   & ---   & ---   & ---   & ---   & ---   & ---   & ---   & ---   & ---   & R16   & ---   & --- \\
    Leicester City & England & ---   & ---   & ---   & ---   & ---   & ---   & ---   & ---   & ---   & ---   & ---   & ---   & ---   & QF    & ---   & ---   & ---   & ---   & ---   & --- \\
    Lille & France & ---   & ---   & ---   & R16   & ---   & ---   & ---   & ---   & ---   & ---   & ---   & ---   & ---   & ---   & ---   & ---   & ---   & ---   & R16   & --- \\
    Liverpool & England & ---   & W     & R16   & F     & SF    & QF    & ---   & ---   & ---   & ---   & ---   & ---   & ---   & ---   & F     & W     & R16   & QF    & F     & R16 \\ \toprule
    \end{tabular}
\begin{tablenotes}
\item
The year corresponds to the beginning of the season, e.g.\ 2003 represents the 2003/04 season.
\item
Abbreviations: R16 = Round of 16; QF = quarterfinals; SF = semifinals; F = final; W = winner.
\end{tablenotes}
\end{threeparttable}
\end{adjustbox}
\end{sidewaystable}

\begin{sidewaystable}
  \centering
  \caption{Clubs in the knockout stage of the Champions League between the seasons 2003/04 and 2022/23 II.}
  \label{Table_A2}
\begin{adjustbox}{max width=\textwidth}
\begin{threeparttable}	
    \rowcolors{1}{}{gray!20}
    \begin{tabular}{ll ccccc ccccc ccccc ccccc} \toprule
    Team  & Country & 2003  & 2004  & 2005  & 2006  & 2007  & 2008  & 2009  & 2010  & 2011  & 2012  & 2013  & 2014  & 2015  & 2016  & 2017  & 2018  & 2019  & 2020  & 2021  & 2022 \\ \bottomrule
    Lokomotiv Moscow & Russia & R16   & ---   & ---   & ---   & ---   & ---   & ---   & ---   & ---   & ---   & ---   & ---   & ---   & ---   & ---   & ---   & ---   & ---   & ---   & --- \\
    Lyon  & France & QF    & QF    & QF    & R16   & R16   & R16   & SF    & R16   & R16   & ---   & ---   & ---   & ---   & ---   & ---   & R16   & SF    & ---   & ---   & --- \\
    M\'alaga & Spain & ---   & ---   & ---   & ---   & ---   & ---   & ---   & ---   & ---   & QF    & ---   & ---   & ---   & ---   & ---   & ---   & ---   & ---   & ---   & --- \\
    Manchester City & England & ---   & ---   & ---   & ---   & ---   & ---   & ---   & ---   & ---   & ---   & R16   & R16   & SF    & R16   & QF    & QF    & QF    & F     & SF    & W \\
    Manchester United & England & R16   & R16   & ---   & SF    & W     & F     & QF    & F     & ---   & R16   & QF    & ---   & ---   & ---   & R16   & QF    & ---   & ---   & R16   & --- \\
    Marseille & France & ---   & ---   & ---   & ---   & ---   & ---   & ---   & R16   & QF    & ---   & ---   & ---   & ---   & ---   & ---   & ---   & ---   & ---   & ---   & --- \\
    Milan & Italy & QF    & F     & SF    & W     & R16   & ---   & R16   & R16   & QF    & R16   & R16   & ---   & ---   & ---   & ---   & ---   & ---   & ---   & ---   & SF \\
    Monaco & France & F     & R16   & ---   & ---   & ---   & ---   & ---   & ---   & ---   & ---   & ---   & QF    & ---   & SF    & ---   & ---   & ---   & ---   & ---   & --- \\
    Napoli & Italy & ---   & ---   & ---   & ---   & ---   & ---   & ---   & ---   & R16   & ---   & ---   & ---   & ---   & R16   & ---   & ---   & R16   & ---   & ---   & QF \\
    Olympiacos & Greece & ---   & ---   & ---   & ---   & R16   & ---   & R16   & ---   & ---   & ---   & R16   & ---   & ---   & ---   & ---   & ---   & ---   & ---   & ---   & --- \\
    Panathinaikos & Greece & ---   & ---   & ---   & ---   & ---   & R16   & ---   & ---   & ---   & ---   & ---   & ---   & ---   & ---   & ---   & ---   & ---   & ---   & ---   & --- \\
    Paris Saint-Germain & France & ---   & ---   & ---   & ---   & ---   & ---   & ---   & ---   & ---   & QF    & QF    & QF    & QF    & R16   & R16   & R16   & F     & SF    & R16   & R16 \\
    Porto & Portugal & W     & R16   & ---   & R16   & R16   & QF    & R16   & ---   & ---   & R16   & ---   & QF    & ---   & R16   & R16   & QF    & ---   & QF    & ---   & R16 \\
    PSV Eindhoven & Netherlands & ---   & SF    & R16   & QF    & ---   & ---   & ---   & ---   & ---   & ---   & ---   & ---   & R16   & ---   & ---   & ---   & ---   & ---   & ---   & --- \\
    Rangers & Scotland & ---   & ---   & R16   & ---   & ---   & ---   & ---   & ---   & ---   & ---   & ---   & ---   & ---   & ---   & ---   & ---   & ---   & ---   & ---   & --- \\
    RB Leipzig & Germany & ---   & ---   & ---   & ---   & ---   & ---   & ---   & ---   & ---   & ---   & ---   & ---   & ---   & ---   & ---   & ---   & SF    & R16   & ---   & R16 \\
    RB Salzburg & Austria & ---   & ---   & ---   & ---   & ---   & ---   & ---   & ---   & ---   & ---   & ---   & ---   & ---   & ---   & ---   & ---   & ---   & ---   & R16   & --- \\
    Real Madrid & Spain & QF    & R16   & R16   & R16   & R16   & R16   & R16   & SF    & SF    & SF    & W     & SF    & W     & W     & W     & R16   & R16   & SF    & W     & SF \\
    Real Sociedad & Spain & R16   & ---   & ---   & ---   & ---   & ---   & ---   & ---   & ---   & ---   & ---   & ---   & ---   & ---   & ---   & ---   & ---   & ---   & ---   & --- \\
    Roma  & Italy & ---   & ---   & ---   & QF    & QF    & R16   & ---   & R16   & ---   & ---   & ---   & ---   & R16   & ---   & SF    & R16   & ---   & ---   & ---   & --- \\
    Schalke 04 & Germany & ---   & ---   & ---   & ---   & QF    & ---   & ---   & SF    & ---   & R16   & R16   & R16   & ---   & ---   & ---   & R16   & ---   & ---   & ---   & --- \\
    Sevilla & Spain & ---   & ---   & ---   & ---   & R16   & ---   & R16   & ---   & ---   & ---   & ---   & ---   & ---   & R16   & QF    & ---   & ---   & R16   & ---   & --- \\
    Shakhtar Donetsk & Ukraine & ---   & ---   & ---   & ---   & ---   & ---   & ---   & QF    & ---   & R16   & ---   & R16   & ---   & ---   & R16   & ---   & ---   & ---   & ---   & --- \\
    Sparta Prague & Czech Republic & R16   & ---   & ---   & ---   & ---   & ---   & ---   & ---   & ---   & ---   & ---   & ---   & ---   & ---   & ---   & ---   & ---   & ---   & ---   & --- \\
    Sporting CP & Portugal & ---   & ---   & ---   & ---   & ---   & R16   & ---   & ---   & ---   & ---   & ---   & ---   & ---   & ---   & ---   & ---   & ---   & ---   & R16   & --- \\
    Stuttgart & Germany & R16   & ---   & ---   & ---   & ---   & ---   & R16   & ---   & ---   & ---   & ---   & ---   & ---   & ---   & ---   & ---   & ---   & ---   & ---   & --- \\
    Tottenham Hotspur & England & ---   & ---   & ---   & ---   & ---   & ---   & ---   & QF    & ---   & ---   & ---   & ---   & ---   & ---   & R16   & F     & R16   & ---   & ---   & R16 \\
    Valencia & Spain & ---   & ---   & ---   & QF    & ---   & ---   & ---   & R16   & ---   & R16   & ---   & ---   & ---   & ---   & ---   & ---   & R16   & ---   & ---   & --- \\
    Villareal & Spain & ---   & ---   & SF    & ---   & ---   & QF    & ---   & ---   & ---   & ---   & ---   & ---   & ---   & ---   & ---   & ---   & ---   & ---   & SF    & --- \\
    Werder Bremen & Germany & ---   & R16   & R16   & ---   & ---   & ---   & ---   & ---   & ---   & ---   & ---   & ---   & ---   & ---   & ---   & ---   & ---   & ---   & ---   & --- \\
    Wolfsburg & Germany & ---   & ---   & ---   & ---   & ---   & ---   & ---   & ---   & ---   & ---   & ---   & ---   & QF    & ---   & ---   & ---   & ---   & ---   & ---   & --- \\
    Zenit Saint Petersburg & Russia & ---   & ---   & ---   & ---   & ---   & ---   & ---   & ---   & R16   & ---   & R16   & ---   & R16   & ---   & ---   & ---   & ---   & ---   & ---   & --- \\ \bottomrule
    \end{tabular}
\begin{tablenotes}
\item
The year corresponds to the beginning of the season, e.g.\ 2003 represents the 2003/04 season.
\item
Abbreviations: R16 = Round of 16; QF = quarterfinals; SF = semifinals; F = final; W = winner.
\end{tablenotes}
\end{threeparttable}
\end{adjustbox}
\end{sidewaystable}

\end{document}